\documentclass[aps, pra,twocolumn,showpacs,superscriptaddress,10pt]{revtex4-1}

\usepackage{amsfonts,amssymb,amsmath}
\usepackage[]{graphics,graphicx,epsfig}
\usepackage{amsthm}
\usepackage{color}

\def\identity{\leavevmode\hbox{\small1\kern-3.8pt\normalsize1}}

\newtheorem{theorem}{Theorem}

\newtheorem{deffi}{Definition}
\newtheorem{propo}{Proposition}
\newtheorem{corollary}{Corollary}
\newcommand{\be}{\begin{eqnarray} \begin{aligned}}
\newcommand{\ee}{\end{aligned} \end{eqnarray} }
\newcommand{\bpr}{\begin{propo}}
\newcommand{\epr}{\end{propo}}
\newcommand{\bpf}{\begin{proof}}
\newcommand{\epf}{\end{proof}}
\newcommand{\ket}[1]{\left | #1 \right\rangle}

\renewcommand{\epsilon}{\varepsilon}

\newcommand{\marcin}[1]{{\color{black} #1}}
\newcommand{\adrian}[1]{{\color{black} #1}}

\usepackage[T1]{fontenc}

\begin{document}

\title{Non-classicality of temporal correlations}

\author{Stephen Brierley}
\affiliation{Heilbronn Institute for Mathematical Research, Department of Mathematics, University of Bristol, Bristol BS8 1TW, UK}

\author{Adrian Kosowski}
\affiliation{Inria and Universit\'e Paris Diderot, LIAFA, Case 7014, 75205 Paris Cedex 13, France}

\author{Marcin Markiewicz}
\email{marcinm495@gmail.com}
\affiliation{Faculty of Physics, University of Warsaw, Pasteura 5, PL-02-093 Warszawa, Poland}
\affiliation{Institute of Theoretical Physics and Astrophysics, University of Gda\'nsk, 80-952 Gda\'nsk, Poland}

\author{Tomasz Paterek} 
\affiliation{School of Physical and Mathematical Sciences, Nanyang Technological University, Singapore}
\affiliation{Centre for Quantum Technologies, National University of Singapore, Singapore}

\author{Anna Przysi\k{e}\.zna}
\affiliation{Institute of Theoretical Physics and Astrophysics, University of Gda\'nsk, 80-952 Gda\'nsk, Poland}
\affiliation{National Quantum Information Centre in Gdansk, Andersa 27, 81-824 Sopot}

\begin{abstract}
{The results of space-like separated measurements are independent of distant measurement settings, a property one might call two-way no-signalling. 
In contrast, time-like separated measurements are only one-way no-signalling since the past is independent of the future but not vice-versa. 
For this reason some temporal correlations that are formally identical to non-classical spatial correlations can still be modelled classically.
We propose a new formulation of Bell's theorem for temporal correlations, namely we define non-classical temporal correlations as the ones which cannot be simulated by propagating in time the classical information content of a quantum system given by the Holevo bound.
We first show that temporal correlations between results of any projective quantum measurements on a qubit can be simulated classically.
Then we present a sequence of POVM measurements on a single $m$-level quantum system that cannot be explained by propagating in time an $m$-level classical system and using classical computers with unlimited memory.
}

\end{abstract}

\maketitle

\emph{Introduction}. The violation of a Bell inequality \cite{Bell64, Clauser69, Brunner14} demonstrates that the outcomes of an experiment have contradicted a set of well defined classical intuitions. Quantum mechanics allows correlations between space-like separated parties that have no explanation in terms of a hidden variable model, i.e. they cannot be reproduced with the help of classical computers running pre-agreed algorithms.
However, when correlations are generated in a temporal scenario, by a sequence of time-like separated measurements, it is more difficult to demonstrate their non-classical nature. The causal structure of physics implies only \emph{one-way} no-signalling, namely the impossibility of sending communication backwards in time. The only bound on forward signalling is the information capacity of the physical system.

Here we analyse a single quantum system measured at $n$ points in time and consider to what extent one can prove that the temporal correlations between these measurement outcomes could \emph{not} be generated by a classical system. We assume an idealization, in which the $m$-level physical system carries no hidden degrees of freedom. In this case the classical information capacity of the system is $\log_2 m$; known as the \emph{Holevo bound} \cite{Holevo73}.

The previous approach to demonstrate non-classicality of temporal quantum correlations are so called ``temporal Bell inequalities'' \cite{LG85, BTCV04, Lapiedra06, B09, AHW10, Fritz10, Markiewicz13, BMKG13,Z10, Budroni14}. One of the problems with this approach, stated in \cite{Z10}, is that the classical assumptions behind the temporal inequalities, which are realism and non-invasiveness, were originally chosen to test the quantumness of a temporal evolution of \emph{macroscopic} quantum systems \cite{LG85}. As such, they do not provide a convincing test of quantumness in the case of a single evolving system. Moreover the operational meaning of the assumptions themselves is still a subject of debate \cite{Maroney14, Lapiedra06, Markiewicz13}.

%

The fact that in the sequential scenario the evolving system caries information has significant consequences \cite{Fritz10}.
Let us first consider the simplest possible case: a single two-level system which undergoes a sequence of two black-box operations at time instances $t_1$ and $t_2$.  
Each black box has an input describing which settings are chosen and an output describing the measurement result.
\marcin{For a quantum implementation in which black boxes perform projective measurements the inputs have the form of unit vectors $\vec a(t_1)$ and $\vec a(t_2)$, and outputs read $\alpha(t_1)=\pm 1$ and $\alpha(t_2)=\pm 1$, respectively. It can be verified \cite{BTCV04}, that the temporal correlation function $\langle \alpha(t_1)\alpha(t_2)\rangle$ equals $\vec a(t_1)\cdot \vec a(t_2)$}. This is (up to the sign) the correlation function that would be generated by two separated parties that share the singlet state. Moreover, it leads to a maximal violation of the temporal CHSH inequality \cite{BTCV04}. 
Can we conclude from these two facts that our system gives rise to non-classical temporal correlations? 
The answer is negative --- instead of a single qubit one can communicate one classical bit which together with black boxes equipped with classically correlated real vectors $\vec \lambda$ implement the Toner-Bacon $1$-bit protocol for simulating the singlet state \cite{Toner03}\footnote{Technically, we have to use Toner-Bacon protocol with one of the input vectors reflected, so as to compensate the global sign difference}. 
Furthermore, beginning the evolution with an arbitrary qubit state any sequence of $n$ projective measurements leads to correlations that factor into pairs of dot products of the consecutive input vectors ~\cite{BTCV04}.
Thus the temporal correlations of $n$ projective measurements on a qubit admit a classical simulation essentially using a sequence of Toner-Bacon protocols.

A similar situation can arise in the case of  multi-point correlations. For example, consider a sequence of three black boxes with two-setting inputs $\phi(t_k)=\{0,\pi/2\}$ for $k=1,2,3$ and binary outputs $\alpha(t_k)=\pm1$ together with the promiss that the inputs fulfill the constraint $\sum_k \phi(t_k)=\{0,\pi\}$. Let us assume, that the correlation function of outputs $\langle \alpha(t_1)\alpha(t_2) \alpha(t_3)\rangle$ equals $\cos(\phi(t_1)+\phi(t_2)+\phi(t_3))$. This is the correlation function of a Greenberger-Horne-Zeilinger (GHZ) state, and in the spatial scenario leads to the GHZ paradox~\cite{GHZ90}. \marcin{In the temporal scenario such a function can be obtained by a sequence of two-outcome POVM measurements on a single qubit \cite{Markiewicz14}.} However, as we prove later in the paper, the entire setup can be simulated by a classical protocol with exactly 1 bit of classical communication. Again, we cannot verify that truly non-classical correlations have appeared.

\emph{Definition of non-classical temporal correlations}. The above examples show that identifying a violation of temporal Bell inequalities with so called 'entanglement in time' \cite{BTCV04} is not always accurate. Indeed, forward signalling of classical information in the sequential scenario leaves a kind of \emph{communication loophole}: correlation functions which would be considered non-classical in the spatial setting can be simulated by classical protocols that use a classical communication channel with capacity given by the Holevo bound of the corresponding quantum particle. 
\marcin{The communication implicitly involved in a sequential process motivates adapting the concept of effective classical simulability with respect to the communication cost (see Fig. \ref{sim}): } 
\begin{deffi}
\label{strongdef}
The temporal correlation function $E(Y_1,\ldots,Y_N|X_1,\ldots,X_N)$ of the $m$ level \marcin{physical} system is non-classical if all classical algorithms that simulate the function require more than $\log_2 m$ bits of classical communication at some step of the simulation. 
\end{deffi}
The idea of quantifying the degree of non-classicality of a physical process \marcin{distributed in space or time} by the amount of classical resources needed to simulate it has already appeared in many contexts: communication complexity of simulating spatial quantum correlations \cite{BCT99, MBCC01, Toner03, RT07, DLR07, SZ08, VB09, Kosowski13, Brassard13}, memory complexity of simulating contextual effects \cite{KGPLC13}, and memory complexity of simulating unitary evolution \cite{Trojek05, Galvao08, Montina08, Z10}. 
\marcin{The fact that simulation of some sequential quantum procedures like contextuality tests or unitary evolution demands resources exceeding the Holevo bound has been already noticed in the literature \cite{Galvao08, KGPLC13}. Our approach generalizes these ideas to the scenario of sequential measurements performed by black boxes, with no restriction on their internal operations or memory. Definition \ref{strongdef} provides a \emph{theory-independent} characterisation of non-classical temporal correlations: one does not need to specify the physical implementation which leads to given correlations, but only the number of degrees of freedom of the physical system entering the boxes. In contrast, in the most similar scenario of quantifying the memory cost of quantum contextuality \cite{KGPLC13} the counted resource is the total number of internal states used by the simulating machines; which can be different from the size of the communication between the steps needed for simulation. We point out that there are other notions of simulability, which can be used to define non-classicality (eg. in terms of computational complexity \cite{Oszmaniec14}). In this work we solely refer to simulability in the context of \emph{communication complexity}}.

\begin{figure}[!b]
\includegraphics[width=0.45\textwidth]{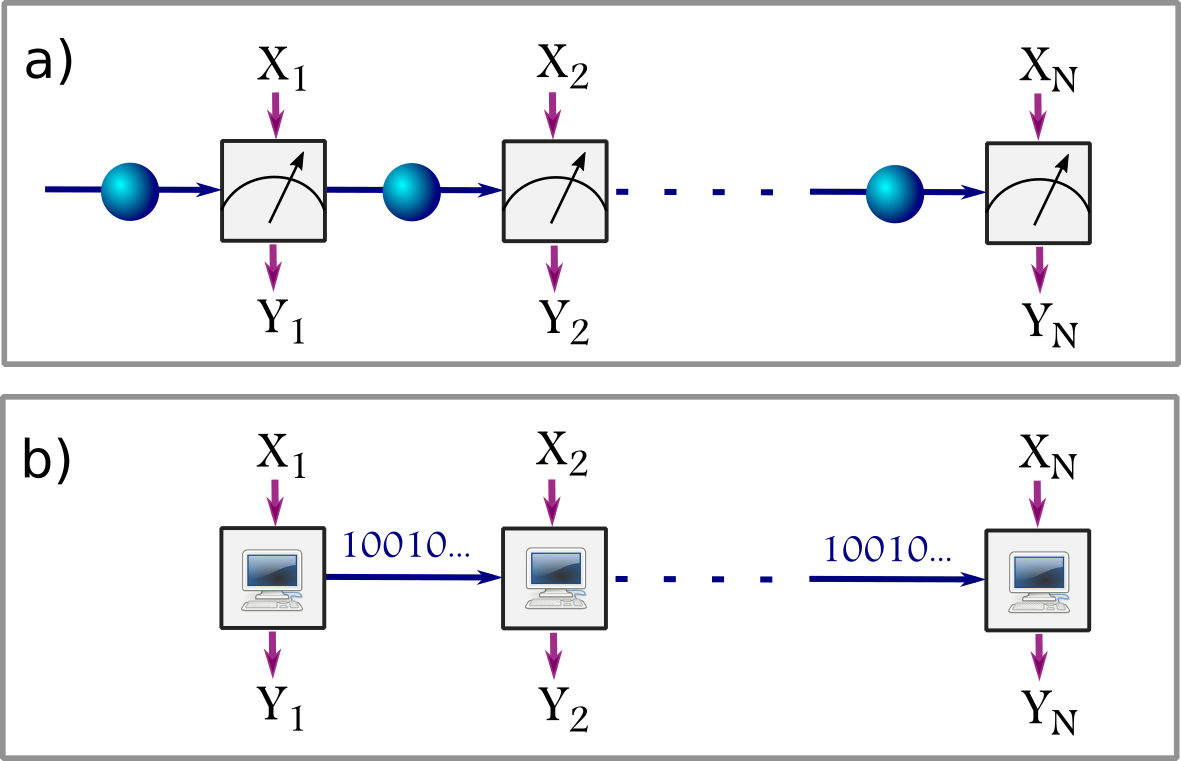}
\caption{Temporal correlation functions. a) A sequence of $n$ consecutive measurements on a single quantum system with settings provided by inputs $X_k$ and outcomes given by numbers $Y_k$. We say that temporal correlations are non-classical if there is no classical simulation depicted in panel b). At the $i$-th step of the simulation, the $i$-th black box can perform local computations and send classical communication to the next box. The correlation function obtained in the scenario a) is non-classical if every classical simulation b) requires more communication than the classical information capacity of the quantum particle in at least one stage of the simulation.}
\label{sim}
\end{figure}

We demonstrate the utility of our definition by presenting a sequence of quantum measurements that give non-classical temporal correlations provided the number of measurements is sufficiently large. The quantum system and measurements we propose have the appealing property that they are within the reach of current experimental techniques. In the simplest case, we find non-classical temporal correlations for a sequence of $16$ POVM measurements on a qubit system.

%
%

\emph{Non-classicality of temporal GHZ correlations}. We show that the temporal GHZ correlations, that is temporal correlations which have the same form as spatial correlations of an $n$-qu$m$it GHZ state:
\begin{equation}
\label{genGHZ}
\ket{GHZ}=\frac{1}{\sqrt{m}}\sum_{i=1}^m\ket{i}^{\otimes n},
\end{equation}
are non-classical on condition that  the number of measurement steps and the number of settings per observer is sufficiently large. To prove this fact we utilize the so called modulo-$(m,d)$ games \cite{Boyer04}. In the spatial scenario, these games are distributed computing tasks that can be solved with certainty with the help of shared GHZ-state \eqref{genGHZ} but not with classical randomized algorithms. We translate these games into the sequential scenario and prove that they can be always solved exactly by a sequence of POVM measurements on a single qu$m$it giving rise to temporal GHZ correlations.
We show that for some sets of parameters $n,m,d$ the correlations cannot be simulated by any protocol whose communicates are less than the number of classical bits given by the Holevo bound.

\begin{deffi}[sequential $n$-point modulo-($m,d$) problem]
A sequential $n$-point modulo-$(m,d)$ problem is a communication complexity task, in which $n$ separate ordered parties are given $(\log d)$-bit inputs $X_k$, with the promise that $\sum_{k=1}^n X_k \operatorname{mod} d=0$. The task of the parties is to output values  $Y_k\in\{0,1,\ldots,m-1\}$ fulfilling $d \sum_{k=1}^n Y_k \equiv \sum_{k=1}^n X_k\operatorname{mod} (md)$ in a sequential protocol, which at $k$-th stage allows $k$-th party to produce her local output $Y_k$ and communicate a $c_k$-bit message $M_k$ to the $(k+1)$-st party.
\end{deffi}
First we prove, that the above problem can be solved with certainty by a sequence of appropriate generalized quantum measurements on a single qu$m$it.
To solve the modulo-$(m,d)$ problem in a spatial domain \cite{Boyer04}, the $n$ parties share the state $\eqref{genGHZ}$ and apply local unitary operations $F_m^\dagger (S_{dm})^{X_k}$ to their particles,
where $\left[F_m\right]_{\alpha\beta}=\frac{1}{\sqrt{m}}\exp{\left(\frac{2i\alpha\beta\pi}{m}\right)}$ and $\left[S_{dm}\right]_{\alpha\beta}=\exp{\left(\frac{2i\beta\pi}{dm}\right)}\delta_{\alpha\beta}$,  $\alpha,\beta={0,1,...,m-1}$.
Finally the states are locally measured in the standard basis. 
Let us now find a temporal counterpart.
First, note that the bond dimension of the matrix product state representation \cite{MPSrev} of the  $n$-qu$m$it GHZ state is $m$, which implies that the state can be constructed by a sequential cascade of two qu$m$it gates $U$ \cite{Schon, BGWVC08}.  
This property implies that one can map arbitrary local projective measurements in the spatial scenario into a sequence of POVM measurements on a single system with dimension $m$ whilst keeping the correlation function fixed \cite{Markiewicz14}. 
Following the construction presented in Ref.~\cite{Markiewicz14} one finds the measurement operators $K_{Y_k}$ corresponding to different outcomes $Y_k$ by solving the system of equations:
\begin{eqnarray}
\label{povm1}
S(U(\ket{\psi}\otimes\ket{0}))&=&\sum_{Y_k=0}^{m-1}\left(K_{Y_k}\ket{\psi}\right)\otimes F_m^\dagger (S_{dm})^{X_j}\ket{Y_k},\nonumber
\end{eqnarray}
where $S$ is a swap operator, and $\ket{\psi}$ is an arbitrary state.
For odd dimensions $m$, measurement operators $K_{Y_k}$ obtained in this way are diagonal matrices $\mathcal{M}$ with $j$th diagonal element equal to $\left[\mathcal{M}(d)\right]_{jj}=\frac{1}{\sqrt{m}}\exp\left(\frac{i X_k \pi (2^j-2)}{d}\right)$.
For even $m$, measurement operators  corresponding to outputs $Y_k$ are $K_{Y_k}=\mathcal{A}_{\rm Y_k} \mathcal{M}$, where $\mathcal{A}_{\rm Y_k}$ is a diagonal matrix with elements  $\pm1$  and  $ \pm  i$.
In particular for $m=2$, $A_0=\mathrm{diag}(1,1)$,  $A_1=\mathrm{diag}(1,-1)$; for $m=4$, $A_0=\mathrm{diag}(1,1,1,1)$,  $A_1=\mathrm{diag}(1,-i,-1,i)$, $A_2=\mathrm{diag}(1,-1,1,-1)$,  $A_3=\mathrm{diag}(1,i,-1,-i)$;
for $m=6$, $A_0=A_2=A_4=\mathrm{diag}(1,1,1,1,1,1)$,  $A_1=A_3=A_5=\mathrm{diag}(1,-1,1,-1, 1,- 1)$.



We now consider classical simulations of modulo-$(m,d)$ games that use at most $\log m$ bits of classical communication at each step. 
Before presenting the main result, we first discuss the simplest case $m=d=2$.
It turns out, that this problem \emph{can} be simulated with a single bit of communication at each stage (see Appendix A).
Since the modulo-$(2,2)$  problem for $n=3$ is equivalent to the scenario of the original GHZ paradox \cite{Gavoile09}, discussed in the introduction, it follows that the two-setting GHZ qubit correlations, which in the spatial domain reveal strong non-classicality, do not fulfill the definition of non-classical temporal correlations. 

We now provide a general lower bound on the amount of communication which is needed to classically solve these problems, which will imply the main result of our work.

\begin{theorem}
\label{modmd}
Every classical protocol which solves the sequential modulo-$(m,d)$ problem with certainty uses \adrian{at least $c_k=\log (d/m)$ bits} of communication in all stages of the protocol except at most $md-1$ (not necessarily consecutive) stages, when $d$ is an integer power of $2$ and $m$ is even.
\end{theorem}
The proof comes by a minor modification of the argument used in~\cite{Galvao08}, \marcin{adapted to the scenario of sequential measurements}. We provide a stand-alone proof for completeness, since the original proof applies an iterated argument in a slightly informal way, leading to incorrect constants in the analysis. To perform a proof by contradiction, fix any classical protocol which claims to solve the modulo-$(m,d)$ problem with certainty, while using \adrian{less than} $c_k$ bits of communication in some $n_0 = md$ stages. We will act as an adversary, constructing two valid inputs of the modulo-$(m,d)$ problem, $X_k$ and $X_k'$ such that $\sum_{k=1}^n X_k \not \equiv \sum_{k=1}^n X'_k \operatorname{mod} md$. The inputs will be defined so that the corresponding outputs $Y_k$ and $Y_k'$ will be indistinguishable, i.e., $Y_k = Y'_k$, for all $1 \leq k \leq n$. Hence, the modulo-$(m,d)$ problem will be solved incorrectly for at least one of the inputs $X_k$, $X_k'$, leading to a contradiction.

The construction proceeds as follows. Let $K_0 = \{k_1, k_2, \ldots, k_{n_0}\}$ be the set of indices of the stages for which the protocol sends messages of size less than $c_k$. We proceed with the construction of inputs $X_k, X_k'$ sequentially, so that the following predicates are fulfilled at any step $k$:
\begin{itemize}
\item For all $j\leq k$, $Y_j = Y'_j$,
\item For all $j\leq k$, $M_j = M'_j$, where $M_j$ is the message.
\end{itemize}
The construction of inputs  $X_k, X_k'$ proceeds as follows:
\begin{itemize}
\item For any stage $k\notin K_0$, $k\neq n$, we set $X_k$ arbitrarily, and put $X_k' = X_k$. Clearly, since $M_{k-1} = M'_{k-1}$ by the inductive assumption and $X_k' = X_k$, the protocol will act identically in the $k$-th step in both cases, thus we have $M_k = M_k'$ and $Y_k = Y'_k$.
\item For any stage $k\in K_0$, given message $M_{k-1} = M'_{k-1}$ we consider the set of outcome pairs $p(x) = (M_k(x), Y_k(x))$ of the execution $k$-th step of the protocol, taken over all possible inputs $x \in \{0,1,\ldots,d-1\}$. Since $|M_k| \leq 2^{c_k}$ for $k\in K_0$ and $Y_k\in\{0,1,\ldots,m-1\}$, the set of possible output pairs $p(x)$ \adrian{has less than $2^{c_k} m$ elements, where we note that $2^{c_k} m = 2^{\log(d/m)} m = d$. Consequently, }there exists a pair of values $x <x'$, $x,x' \in \{0,1,\ldots,d-1\}$, such that $p(x) = p(x')$. We denote $x' = x + \Delta_k$, with $\Delta_k\in \{1,2,\ldots, d-1\}$. We now put $X_k = x$, and choose $X_k'\in \{X_k, X_k + \Delta_k\}$, according to a rule which will be described later. Regardless of this choice, we have $M_k = M_k'$ and $Y_k = Y'_k$.
\item Finally, in stage $k=n$, we set $X_k$ so that the input $\{X_k\}_{k=1}^n$ satisfies the modulo-$d$ promise, and also put $X_n' = X_n$.
\end{itemize}
It remains to show that it is possible to fix $X_k'$ from among each pair of considered values $\{x, x+\Delta\}$ for $k\in K_0$, so that $\sum_{k\in K_0} X_k \not\equiv \sum_{k\in K_0} X'_k \operatorname{mod} (md)$. This is possible by the following lemma, whose proof is presented in Appendix B.

\emph{Lemma 1}. Let $\{\Delta_k\}_{k \in K_0}$ be any sequence of integers, with $\Delta_k \in \{1,2,\ldots, m d -1\}$, where $m$ is even and $d=2^s$ for some integer $s>0$. Then, there exists a subset of indices $K_0' \subseteq K_0$ such that $\sum_{k\in K_0'} \Delta_k \equiv 0 \operatorname{mod} d$ and $\sum_{k\in K_0'} \Delta_k \not\equiv 0 \operatorname{mod} (md)$.

Using Lemma 1, we then pick $X_k' = X_k + \Delta_k$ for all steps $k\in K_0'$, and put $X_k' = X_k$ for all steps $k \in K_0\setminus K_0'$. This completes our construction.

\marcin{In the above proof, we restricted our considerations to deterministic protocols. For randomized protocols, the claim of the proposition also holds in the following sense: any randomized protocol which does not satisfy the assumptions of the proposition will lead to an incorrect output for some instances of the modulo-$md$ problem, with strictly positive probability.}

The above theorem can be treated as a temporal version of Bell inequalities with auxiliary communication in the spatial scenario \cite{Bacon03, Maxwell14}, and directly leads to the main result of our work:

\begin{propo}
\label{mainprop}
The temporal GHZ correlations arising from a sequential measurements on a single qu$m$it, where $m$ is even, are non-classical for \adrian{$n \geq 2m^3$}.
\end{propo}
\begin{proof}
It suffices to show that there exists a modulo-$(m,d)$ game for some $d$ and $n$, for which classical simulation uses in at least one stage of the protocol more than $\log m$ bits of communication. 
Using Theorem 1, we need to choose parameters so that the following two conditions are fulfilled:
\begin{itemize}
	\item \adrian{$\log(d/m)>\log m$}, which means that the classical communication needed is greater than the Holevo bound for the system,
	\item \adrian{$n\geq md$}, which guarantees that the amount of communication equal to \adrian{$\log(d/m)>\log m$} bits is needed in at least one stage.
\end{itemize}
\adrian{Now, let $d$ be the smallest integer power of $2$ larger than $m^2$, we have $d\leq 2m^2$.
Taking any $n \geq 2m^3 \geq md$, guarantees that the numbers $m,n,d$ fulfill both of the above conditions.}
\end{proof}

The above proposition shows that temporal GHZ correlations of any qu$m$it reveal temporal non-classicality, if $n$ is sufficiently large, which implies that any simulating protocol uses more bits of communication than the Holevo bound in at least one stage of the protocol. As a matter of fact, our theorem shows that this actually happens in almost all stages of the protocol (for detailed analysis see Appendix C).

\marcin{We point out that the above proposition holds also for a single qubit case and its possibility follows from allowing POVM measurements. Therefore we provide a first demonstration, that simulation of a \emph{temporal correlation function} on qubit demands resources that exceed the Holevo bound (note that a similar effect for a \emph{unitary evolution} of a single qubit was shown before \cite{Galvao08}).}

\emph{Conclusions.} In general terms, correlations between physical systems can be considered in two distinct scenarios: spatial and temporal. In the spatial scenario local measurements are performed by space-like separated parties who may share a source of joint randomness but who are unable to communicate. Temporal correlations, on the other hand arise from a sequence of measurements on a single physical system at different time instances. Communication is now allowed from one time instance to the next but is limited by the information capacity of the system.

We showed that in the temporal measurement scenario one can define a notion of non-classical $n$-point correlations with a clear operational interpretation. Namely, such correlations cannot be simulated by any classical protocol whose communication is limited by the  Holevo capacity of the evolving quantum system. In addition, we demonstrated that the temporal analogue of generalised GHZ correlations arising from sequential measurements on a single qu$m$it, reveal non-classicality in the temporal scenario provided the number of measurements is large enough. Apart from these foundational issues, we provided the first general lower bound on the communication complexity of simulating multi-point quantum correlations in a sequential measurement scenario (see Ref. \cite{Brassard13} for results on the classical communication cost of simulating spatial GHZ correlations).


\emph{Acknowledgements.}---We would like to thank Marcin Paw{\l}owski, Rafa{\l} Demkowicz-Dobrza\'nski and Marek \.Zukowski for helpful discussions. MM and AP gratefully acknowledge the financial support by Polish Ministry of Science and Higher Education Grant no. IdP2011 000361. This work is supported by the National Research Foundation, Ministry of Education of Singapore Grant No. RG98/13, start-up grant of the Nanyang Technological University,  the NCN   Grant   No.
2012/05/E/ST2/02352 and by the EC under the FP7 IP project SIQS co-financed by the Polish Ministry of Science and Higher Education.

\appendix

\section{Communication complexity of sequential modulo-(2,2) problem}

\begin{corollary}
\label{smod4sim}
The sequential modulo-$(2,2)$ problem can be solved with certainty with one bit of (classical) communication at each stage of the protocol (that is $c_k=1$ for all $k=1,\ldots,N-1$).
\end{corollary}
\begin{proof}
Recall that in the sequential modulo-$(2,2)$ problem, the parties are given bit inputs $X_k\in\{0,1\}$ with the promise $\sum_i X_i\equiv 0 \operatorname{mod} 2$ and output $Y_k\in\{0,1\}$ so that the outputs fulfill $\sum_i 2Y_i-\sum_i X_i\equiv 0 \operatorname{mod} 4$. Let us consider the following deterministic protocol:
\begin{itemize}
\item The message is initialised to $M_0=0$,
\item if $X_i= M_{i-1} =1$, then the $i$-th party returns $Y_i=1$, and otherwise it returns $Y_i=0$,
\item the $i$-th party sends $M_i=(X_i+M_{i-1}) \mod 2$ to the $(i+1)$-st party.
\end{itemize}
The protocol works due to the property: $\sum_{i\leq j} X_i=M_j+\sum_{i\leq j} 2Y_i$, which can be shown by induction on $j$. Since $M_j \equiv \sum_{i\leq j} X_i \operatorname{mod} 2$, by the promise on the input for $j=N$ we have $M_N = 0$. Hence, $\sum_{i\leq N} X_i=\sum_{i\leq N} 2Y_i$, and the claim follows.
\end{proof}
The protocol is valid because the expectation value of $\sum_{k=1}^n Y_k$ for settings that satisfy the promise (assuming $X_k$ corresponds to $\hat x$ and $\hat y$ directions on the Bloch sphere) is either equal to $0$ (for even number of pairs of settings equal to 1) or 1 (for odd number of the pairs). It is therefore sufficient to keep track of the parity of the number of settings equal to $1$. This is exactly acomplished by the protocol. Note that the "local" (single-point) expectation values are not modelled by this protocol.

\section{Proof of Lemma 1}

W.l.o.g.\ let $K = \{1,\ldots,n\}$ within the proof of this claim. Define $S_i$ be the set of all modulo-$md$ remainders which can be obtained using subset sums of the first $i$ elements, 
$$ S_i = \{\left(\sum_{j\in I} \Delta_j\right) \operatorname{mod} (md) : I \subseteq \{1,\ldots,i\}\} \, .$$ 
Consider the sequence of sets, $S_1, S_2, \ldots S_n$. Since $S_i \subseteq S_{i+1} \subseteq {1, \ldots, md-1}$ for $i=1 \ldots n$ and $n \geq md$, we must have that $S_a=S_{a+1}$ for some $a$. The element $\Delta_{a+1} \in S_a$ and the set $S_a$ is invariant with respect to a modulo-$md$ shift by $\Delta_{a+1}$, 
\begin{equation}
S_a = \Delta_{a+1} + S_a \operatorname{mod} md \, . \label{invDelta}
\end{equation}
Let $p$ be the unique odd integer such that $\Delta_{a+1} = p \cdot 2^r$, for some integer $r\geq 0$.
Equation (\ref{invDelta}) implies that all multiples of $\Delta_{a+1}$ are also in $S_a (\operatorname{mod} md)$ and in particular, $2^{s-r} \Delta_{a+1} = p d \operatorname{mod} md \in S_a$. Since $p$ is odd and $m$ even, $2^{s-r} \Delta_{a+1} \neq 0 \operatorname{mod} md$ and $2^{s-r} \Delta_{a+1} = 0 \operatorname{mod} d$, as required.

\section{Communication complexity properties of sequential protocols simulating GHZ correlations}
\begin{corollary}
Any classical protocol simulating temporal GHZ correlations of a single qu$m$it on $n$ parties must:
\begin{enumerate}
\item Send $\Omega(\epsilon \log n)$ bits of communication to the next party for each of at least $n - O(n^\epsilon)$ parties, for any $\epsilon >0$.
\item Contain a sequence of $\Omega(n^{1-\epsilon})$ consecutive parties, each of which needs to send $\Omega(\epsilon \log n)$ bits of communication to the next party, for any $\epsilon >0$.
\end{enumerate}
\end{corollary}
\begin{proof}
It suffices to take $d=\Theta(n^{\epsilon})$ in Theorem 1.
\end{proof}
Note that in order to obtain a violation of the Holevo bound for almost all parties, we need to choose an appropriate value of $n = m^{\Omega(1/\epsilon)}$, so that the required amount of communication of $\Omega (\epsilon \log n)$, which follows from the above proposition, exceeds $\log m$.

\bibliographystyle{apsrev4-1}
\bibliography{ccbib}

\end{document}